\documentclass[a4paper,12pt]{article}

\usepackage[english]{babel}
\usepackage{authblk}
\usepackage{dsfont}
\usepackage{amsfonts}
\usepackage{mathrsfs}
\usepackage{amssymb}        
\usepackage{amsmath}
\usepackage{graphicx,caption}
\usepackage{float}
\usepackage[a4paper]{geometry} 
\usepackage{verbatim}
\usepackage{pstricks}
\usepackage{amsthm}
\usepackage{mathrsfs}

\newtheorem{Theorem}{Theorem}

\newtheorem{Corollary}{Corollary}
\newtheorem{Lemma}{Lemma}

\newtheorem{Remark}{Remark}

\newcommand{\R}{\mathbb{R}}

\newcommand{\Z}{\mathbb{Z}}
\newcommand{\N}{\mathbb{N}}

\newcommand{\ed}{\mathrm{e}}

\begin{document}

% title
\title{Equivalence of ensembles, condensation and glassy dynamics in the Bose-Hubbard Hamiltonian}

\author[1]{Fran\c cois Huveneers}
\affil[1]{
Ceremade,
UMR-CNRS 7534, 
Universit\'e Paris Dauphine, 
PSL Research University, 
%Place du Mar\'echal de Lattre de Tassigny, 
75775 Paris cedex 16, 
France
}

\author[2]{Elias Theil}
\affil[2]{
Vienna University of Technology, 
1040 Vienna, 
Austria
}

\date{\today}
\maketitle 

\begin{abstract}
\noindent
We study mathematically the equilibrium properties of the Bose-Hubbard Hamiltonian in the limit of a vanishing hopping amplitude.
This system conserves the energy and the number of particles. 
We establish the equivalence between the microcanonical and the grand-canonical ensembles 
for all allowed values of the density of particles $\rho$ and density of energy $\varepsilon$. 
Moreover, given $\rho$, we show that the system undergoes a transition as $\varepsilon$ increases, 
from a usual positive temperature state to the infinite temperature state where a macroscopic excess of energy condensates on a single site.
Analogous results have been obtained by S.~Chatterjee \cite{chatterjee_2017} for a closely related model.
We introduce here a different method to tackle this problem, hoping that it reflects more directly the basic understanding stemming from statistical mechanics. 
We discuss also how, and in which sense, the condensation of energy leads to a glassy dynamics.  
\end{abstract}

\pagebreak

\section{Introduction}
Condensation is a process that occurs spontaneously under appropriate thermodynamical conditions. 
For example, in a Bose-Einstein condensate, 
a macroscopic fraction of non-interacting bosons occupy the ground-state at sufficiently low temperature, see e.g.~\cite{lieb_et_al_2005}. 
Condensation into a single state is not a genuinely quantum effect though. 
Indeed, real-space condensation takes place for example in the zero-range process, a lattice gas with a stochastic dynamics, at high enough density,
see \cite{evans_hanney_2005,majumdar_2010} for reviews.
However, the jump rates of this process do not respect the detailed balance condition, 
and it is natural to seek for examples within the realm of equilibrium statistical mechanics. 
In this paper, we study the real-space condensation of the energy in the Bose-Hubbard system.
The same phenomenon in a closely related model, the discrete non-linear Schr\"odinger (DNLS) chain, was analyzed through physical arguments in \cite{rumpf_2004},
and at the mathematical level of rigor in \cite{chatterjee_2017} (see also \cite{chatterjee_kirkpatrick_2012,chatterjee_2014} as well as \cite{nam_2018}).  
The Hamiltonian in a volume $\Lambda \subset \Z^d$ is given by 
\begin{equation}\label{eq: Bose Hubbard Hamilotnian}
	H 
	\; = \; 
	H_{pot} + H_{kin}
	\; = \;
	U \sum_{x \in \Lambda} n_x^2 \; + \; J \sum_{x,y\in \Lambda, x\mathsf v y} \left( a_x^\dagger a_{y} + a_{y}^\dagger a_x \right) 
\end{equation}
with $n_x = a_x^\dagger a_x$ the number of bosons at site $x$, 
with $U>0$ the interaction strength and $J\ge 0$ the hopping amplitude, 
and where $x \mathsf v y$ means that $x$ and $y$ are nearest neighbors. 

Our eventual goal is to describe the equilibrium states and the emergence of condensation in the strongly interacting regime $U \gg J$ at positive temperature. 
Below, in this introduction, we first describe the full phase diagram at $J = 0$ and understand heuristically the origin of the condensation.
Next, we argue that most of our conclusions carry over to the case $J>0$ and we analyze the dynamical consequences. 
As it turns out, the time needed to form a condensate diverges very fast as the volume increases and, since this corresponds also to the time to reach equilibrium, 
the dynamics qualifies in this sense as a glass. 
Finally, we comment on the specificity of our approach and on the relation with previous mathematical works.   
In Section~\ref{sec: model and results}, we state precisely our mathematical results, all dealing with the case $J = 0$. 
The proofs are gathered in Sections \ref{sec: equivalence of ensembles} and \ref{sec: concentration on a single site}.

\bigskip
\textbf{Condensation at $J = 0$.}
In this limit, the Hamiltonian $H$ in \eqref{eq: Bose Hubbard Hamilotnian} reduces to a classical Hamiltonian with repulsive on-site interaction, 
and we may choose energy units so that $U = 1$. 
There are two extensive conserved quantities:
the total energy $H$ and the total number of particles $\mathcal N = \sum_{x \in \Lambda} n_x$. 
It is convenient to work with the corresponding densities $\varepsilon = H/V$ and $\rho = \mathcal N/V$, where $V = |\Lambda|$. 
Let us fix some $\rho > 0$ and let us describe the equilibrium state as $\varepsilon$ increases. 
First, for too small $\varepsilon$, there are simply no states in the system 
since any configuration of particles with $\rho > 0$ has also a minimal energy density $\varepsilon_{\mathrm{gs}}(\rho)$.
This ground state energy density is obtained by minimizing fluctuations, i.e.\@ considering configurations where the density of particles is most homogeneous. 
This yields: 
\begin{equation}\label{eq: ground state energy density}
	\varepsilon_{\mathrm{gs}}(\rho) 
	\; = \; 
	(1-(\rho - \lfloor \rho \rfloor)) \lfloor \rho \rfloor^2 + (\rho - \lfloor \rho \rfloor) (\lfloor \rho \rfloor + 1)^2  \, ,
\end{equation} 
where $\lfloor x \rfloor$ denotes the integer part of $x$.
Second, as $\varepsilon$ increases close above the ground state energy density, 
the equilibrium states of the system are described by a usual Gibbs state $\ed^{- \beta (H - \mu \mathcal N)} /Z$, 
with a positive temperature $T = (k_B \beta)^{-1}$ and chemical potential $\mu$
that can be determined by fixing the average value of the number of particles and the energy.
However, as $\varepsilon$ increases further, particles need to pile up in some places, 
since this is the only way to create a large amount of energy with few particles, 
and the repulsive potential becomes thus effectively attractive in this regime.
So, if $\varepsilon > \varepsilon_{\mathrm c} (\rho)$ for some $\varepsilon_{\mathrm c} (\rho)$, 
the micro-canonical entropy density $s_{\mathrm{mic}}$ of the system reduces as the energy increases, 
i.e.\@ $\beta  = \partial s_{\mathrm{mic}} / \partial \varepsilon \le 0$. 
Since the energy grows quadratically with the number of particles per site, 
it is impossible to set $\beta < 0$ in the Gibbs state, and the only reasonable possibility seems thus to be $\beta = 0$ for $\varepsilon > \varepsilon_{\mathrm c} (\rho)$.
Let us notice that $\partial s_{\mathrm{mic}}/ \partial \varepsilon = 0$ still allows for the total entropy to decrease, 
as long as the changes remain sub-extensive, a scenario that we will validate below. 
Incidentally, it is now also possible to guess the value of $\varepsilon_{\mathrm c} (\rho)$. 
Indeed, parametrizing the Gibbs state with $\nu = \beta \mu$ instead of $\mu$, and treating $\nu$ and $\beta$ as independent parameters, 
one can compute explicitely the energy density as a function of the particle density at $\beta=0$: 
\begin{equation}\label{eq: critical density}
	\varepsilon_{\mathrm c} (\rho) \; = \; 2 \rho^2 + \rho . 
\end{equation}
The full phase diagram, validated by our rigorous results, is shown on figure \ref{fig: phase diagram}.

\begin{figure}[t]
    	\centering
   		\includegraphics[draft=false,height = 8cm,width = 12cm]{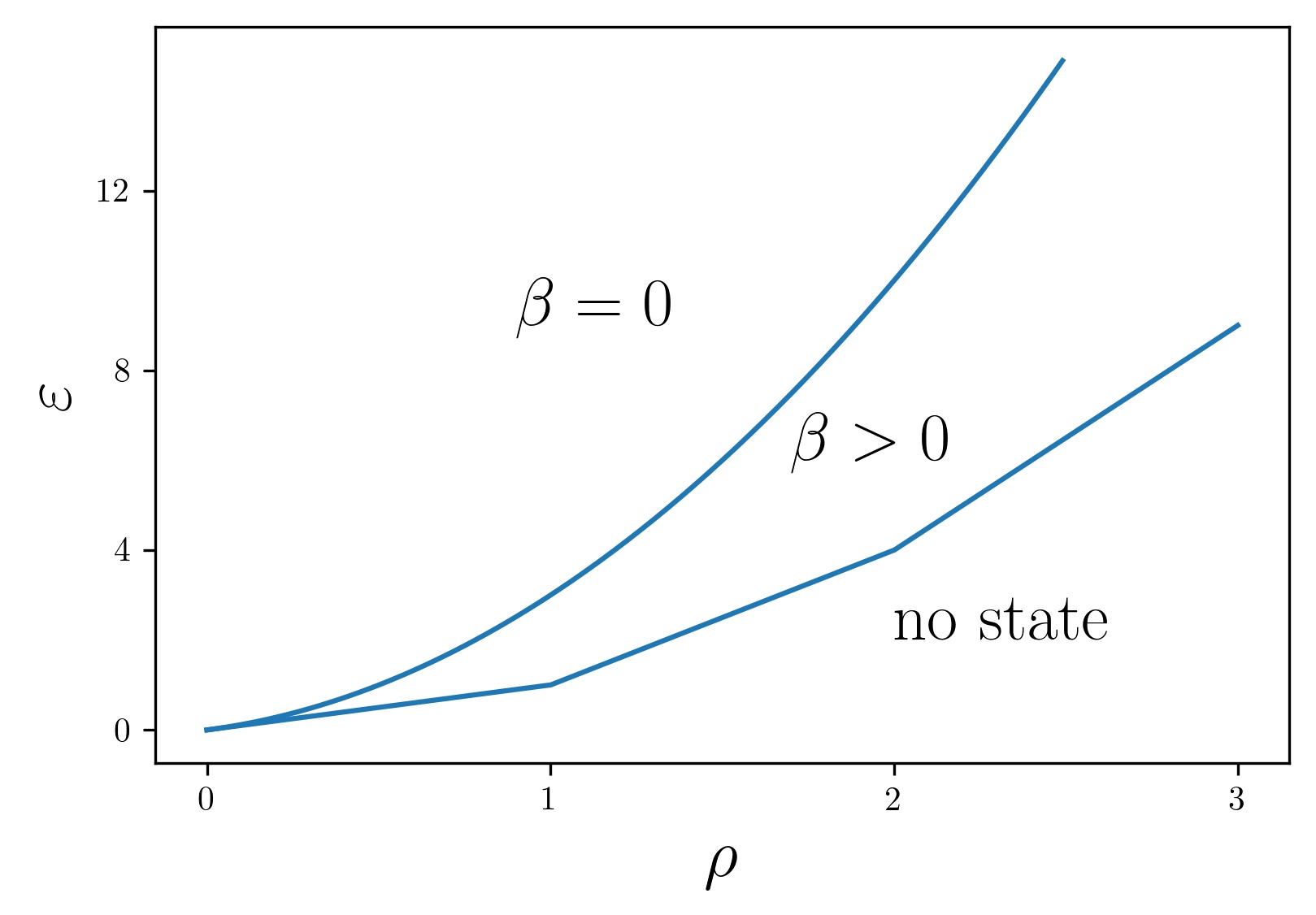}
		\captionsetup{width=.9\linewidth}
    	\caption{Phase diagram of the Hamiltonian $H$ in \eqref{eq: Bose Hubbard Hamilotnian} for $J = 0$.
		The infinite temperature phase $R_{\mathrm{inf}}$ where $\beta = 0$ 
		(assuming the Gibbs state being parametrized by the independent parameters $\nu$ and $\beta$, see the main text),
		is separated from the positive temperature phase $R_{\mathrm{pos}}$ where $\beta >0$
		by the line of critical energy density \eqref{eq: critical density}. 
		The positive temperature phase is bounded from below by the ground state line \eqref{eq: ground state energy density}.} 
    	\label{fig: phase diagram}
\end{figure}

The behavior of the system may seem paradoxical inside the infinite temperature phase since $\varepsilon > \varepsilon_c(\rho)$, 
and we may wonder wether a description of the system by a Gibbs state is actually accurate. 
In its basic form, 
the principle of equivalence of ensembles asserts that the expectation of local and bounded observables in the microcanonical ensemble 
approaches their grand-canonical expectation as the thermodynamic limit $V \to \infty$ is taken. 
In a simple case like ours, at $J = 0$, 
the classical reasoning leading to this conclusion rests only on
the hypothesis that the microcanonical entropy scales as $S_{\mathrm{mic}}(\mathcal N,H,V) \simeq V s_{\mathrm{mic}}(\rho,\varepsilon)$ 
in leading order in $V$ as $V \to \infty$, where the entropy density $s_{\mathrm{mic}}$ depends smoothly on its arguments. 
See \cite{callen}, as well as \cite{touchette_2011,touchette_2015} for a more comprehensive theory. 
There is no reason to doubt the validity of this reasoning in our case, and in Corollary~\ref{co: equivalence of ensembles local observables} below, 
we show indeed the equivalence between the micro-canonical and grand-canonical ensembles for local bounded observables. 
To reach this conclusion, we actually prove the validity of the above scaling for the microcanonical entropy, 
see Theorem~\ref{th: equivalence of ensembles relative entropy} and Remark~\ref{rem: convergence of the microcanonical entropy}. 
In particular, we show that indeed $\partial s_{\mathrm{mic}}/ \partial \varepsilon = 0$ for $\varepsilon > \varepsilon_{\mathrm c} (\rho)$. 

For this to be concretely realized, 
the excess of energy $V\Delta \varepsilon = V (\varepsilon_c - \varepsilon_c (\rho))$ 
must be concentrated on a proportion of sites that vanishes in the thermodynamic limit.
The key observation is that only about $(V \Delta \varepsilon)^{1/2}$ particles on a single site yield the macroscopic quantity of energy $V \Delta \varepsilon$ on this site. 
Hence, concentrating the excess of energy on one site does not change the particles density in the bulk nor the entropy density (in the thermodynamic limit), 
and allows for the prescribed energy density $\varepsilon$ to be realized.
This is thus a possible scenario but, using the same reasoning, 
one may obviously also conclude that the excess of energy could be concentrated on a few sites rather than a single one. 
To understand why the latter typically does not happen, one must analyze sub-extensive variations of the entropy that are no longer visible in the thermodynamic limit.
The basic reasoning is that concentrating the excess of energy $V \Delta \varepsilon$ on several sites requires more particles, 
hence the number of particles in the bulk gets lower, hence the entropy is smaller.  
To make this idea precise from a mathematical point of view, 
we will derive an upper bound on the entropy of all states with maximal number of particles per site being sufficiently small compared to $(V \Delta \varepsilon)^{1/2}$
(note that $\alpha (V \Delta \varepsilon)^{1/2}$ for any $\alpha < 1$ is here considered small). 
This is achieved by constructing the negative temperature Gibbs measure on this set of states 
(there is no trouble in taking $\beta < 0$ if the number of particles per site is bounded above). 
Gibbs states have maximal entropy, and the negative temperature will be determined explicitly with high enough precision so that the entropy
can be accurately bounded as a function of the maximal number of particles per site.  
This is the main step in proving Theorem~\ref{th:single site concentration} below, 
establishing the concentration of the excess of energy on a single site with overwhelming probability.

\bigskip
\textbf{Dynamical glass.}
We now turn to the discussion on the dynamical consequences of the condensation of the energy, but we will not attempt to make any mathematical statements on this.
To get an interesting dynamics, we need obviously to allow for some hopping and so we set here $J > 0$. 

Our first task is to argue that the main conclusions drawn at $J=0$ carry over to the case $J > 0$ as long as $J/U \ll 1$. 
The basic understanding for condensation remains the same:  
Exact diagonalization yields the bound $\| H_{kin} \| \le 2 J \rho V$ and hence, as $\varepsilon$ increases for given $\rho$, 
the potential energy has to increase so that the entropy ceases to grow starting from some threshold. 
Actually, it is possible to be more explicit here. 
Indeed, the quantum infinite temperature state $\langle A \rangle_\nu = \mathrm{Tr}(A \ed^{- \nu \mathcal N})/ \mathrm{Tr} (\ed^{-\nu \mathcal N})$, 
where $A$ is any observable, 
is no different from the classical infinite temperature state.
Since moreover $\langle H_{kin} \rangle_\nu = 0$,
the relation \eqref{eq: critical density} relating the energy density to the density of particles at infinite temperature, 
remains valid and we conjecture that it still yields the critical energy density above which the system forms a condensate. 
This is in line with the recent result in \cite{cherny_et_al_2019} 
(see also \cite{rasmussen_et_al_2000} for analogous considerations in a related classical model), 
where it is shown that all finite temperature states are such that $\varepsilon < \varepsilon_c(\rho)$. 

For $\varepsilon > \varepsilon_c (\rho)$, it is now desirable to obtain a representation of the eigenstates of $H$ as a superposition of the eigenstates of $H_{pot}$, 
i.e.\@ the classical energy states studied previously.  
For this, we admit that the system at $J>0$ is thermalizing, even though the relaxation time towards equilibrium may diverge fast with the volume as we stress below.
Besides this being the most reasonable default option, it is also suggested by recent numerical results in the DNLS chain \cite{mithun_et_al_2018}.
Assuming this, we expect all classical states at a given energy to be fully hybridized. 
Thus, if an eigenstate of $H$ is a superposition of classical states with potential energy in some energy band, 
a statement that holds typically in the micro-canonical ensembles at $J=0$ in this band, becomes valid for the individual eigenstate at $J > 0$. 
For this to be useful, it remains thus to determine the location and the width of the potential energy band. 
Let $|E \rangle$ be an eigenstate of $H$ at an energy density $\varepsilon = E/V > \varepsilon_c(\rho)$. 
First, quantum typicality, or the Eigenstate Thermalization Hypothesis \cite{deutsch_1991,srednicki_1994}, yields
$\langle E | H_{kin} | E \rangle = \langle H_{kin} \rangle_\nu = 0$.
The use of the grand-canonical ensemble here is a priori questionable 
since it does not correctly reproduce the value of all local unbounded observables when $\varepsilon > \varepsilon_c (\rho)$.
However $H_{kin}$ is only linear in the $(n_x)_x$ field, and we only expect inaccuracies for operators that are at least quadratic.
We conclude thus that the band is centered around the potential energy $E$. 
Second, we expect $\langle E | H_{kin}^2 | E \rangle$ to scale as $J^2 V$, hence a width scaling as $J V^{1/2}$. 
So, in particular, the eigenstate $|E \rangle$ is a superposition of classical states at the energy density $\varepsilon = E/V$, 
and it must display an energy condensate. 

Let us now move to the glassy behavior. 
Roughly speaking, a characteristic feature of a glass is the extremely long time needed to evolve towards the thermal equilibrium. 
Such a very slow behavior was observed numerically in \cite{iubini_et_al_2013,mithun_et_al_2018} for the DNLS chain. 
For the Bose-Hubbard Hamiltonian at $\varepsilon > \varepsilon_c (\rho)$, 
there is a natural way to prepare the system out of equilibrium so that the relaxation time diverges very fast as the volume increases.
Indeed, let us fix some length $\ell$ and let us divide the volume $\Lambda$ in cubes of linear size $\ell$. 
We imagine that $\ell$ is large but kept fixed as $V \to \infty$. 
As an initial state, we consider a classical state such that, in each cube, the density of particles is $\rho$ and the energy density is $\varepsilon$. 
This state is actually far from equilibrium, since the entropy density is strictly smaller than the maximal entropy density at infinite temperature. 
As we assume that a condensate features in all eigenstates at this energy density, the system needs to create one in order to reach thermal equilibrium. 
Let us analyze the time needed for this process to happen. 
It is essentially determined by the time required to add the ``last particle" on top of the condensate, 
i.e.\@ moving from about $(V \Delta \varepsilon )^{1/2} - 1$ to $(V \Delta \varepsilon)^{1/2}$ particles on a site. 
If we imagine that the extra particle is extracted from the usual infinite temperature states, 
this process results in an energy difference of order $\Delta E = 2 V^{1/2}\Delta \varepsilon$. 
Releasing such a quantity of energy occurs at a rate \cite{abanin_et_al_2015}
$$
	\tau^{-1} \; \sim  \; \ed^{- \Delta E / \varepsilon_c (\rho)}  \; \sim \; \ed^{- c V^{1/2}}. 
$$ 
Obviously, we could also imagine that such a slow process could be avoided 
by creating a bunch of smaller condensates nearby so that there would be no need to transfer at once such a big amount of energy. 
However, for thermodynamical reasons, it would then take an even longer time to reach such extremely rare configurations.
The time to reach equilibrium grows thus as a stretch exponential in the volume. 
 
Let us make two final remarks on the glassy dynamics. 
First, for many observables of interest, the thermalization time turns out to be much faster than predicted above. 
For example, in $d\ge 2$, we expect that the transport properties are not at all affected by the formation of a condensate, 
and that the conserved quantities evolve on diffusive time scales. 
Second, the slow behavior described above is only remotely connected to the drastic slowing down of conductivity in the Bose-Hubbard chain as $J$ approaches $0$, 
see \cite{bols_de_roeck_2018} as well as \cite{de_roeck_huveneers_2014,de_roeck_huveneers_2015} for similar results in related models.
Indeed, in the latter case, the phenomenon is dynamical rather than entropic, does not depend on the presence of two conserved quantities,
and affects directly the transport properties.  

\bigskip
\textbf{Mathematical methods.}
The main aim of this paper is to prove rigorous statements on the phase diagram at $J = 0$, with methods that mirror directly the understanding developed above. 
Equivalence of ensembles and condensation on a single site have been shown mathematically for the stationary measures of the zero-range process, 
see \cite{grosskinsky_et_al_2003,grosskinsky_2008,chelboun_grosskinsky_2014}. 
The proof of the equivalence of ensembles in \cite{grosskinsky_et_al_2003} is nice and rather straightforward,
but is based on a local central limit theorem as is actually the proof in \cite{chatterjee_2017}. 
We do not know whether such a result is valid for our set-up 
(the parity constraint in the definition \eqref{eq: microcanonical set} below indicates that the statement should be rather specific).
Instead, we exploit the huge degeneracy of the spectrum of $H_{pot}$ to construct explicitly a set of states in the micro-canonical ensemble with 
an entropy that approaches the canonical entropy of the Gibbs state in the thermodynamic limit,  
see Lemma~\ref{lem: discrete approximation Gibbs state} below. 

The way the condensation on a single site is derived in \cite{grosskinsky_et_al_2003} for the zero-range process does not carry over to our system
since the energy per site has still too light tails in the infinite temperature state. 
The system analyzed in \cite{chatterjee_2017}, for which condensation on a single site is shown,
differs from ours by the fact that the variables $n_x$ are continuous.
It is well conceivable that the essence of the arguments used there could be adapted to our set-up. 
Instead, to stick to the argumentation developed in this introduction, our proof proceeds through the introduction of an effective negative temperature state  
as mentioned above and as detailed in Section~\ref{sec: concentration on a single site} below. 
This yields an alternative approach to this problem. 

%the possibility of condensation on more than one states was recently shown in \cite{lukkarinen_2018} for the Berlin-Kac model.

\section{Model and results}\label{sec: model and results}

Let $\N$ be the set of non-negative integers. 
Let $V \in \N \backslash \{ 0 \}$ be the volume and $\Lambda = \{ 0, \dots , V-1 \}$ be the set of sites. 
Let us denote the configuration space by $\Gamma_V = \N^V$. 
Given $\gamma = (\gamma(x))_{x \in \Lambda} \in \Gamma_V$, we associate the particle density and energy density as
$$
	\boldsymbol\rho (\gamma) \; = \;  \frac1V \sum_{x\in \Lambda} \gamma(x), 
	\qquad 
	\boldsymbol\varepsilon (\gamma)  \; = \; \frac1V \sum_{x\in \Lambda} (\gamma(x))^2.
$$

Let us next define the microcanonical and the grand-canonical ensembles. 
First, given a volume $V$, a density of particles $\rho>0$ and an energy density $\varepsilon>0$, 
let 
\begin{equation}\label{eq: microcanonical set}
	\Gamma_{\rho,\varepsilon,V} 
	\; = \;  
	\left\{ 
	\gamma \in \Gamma_V :\quad 
	\boldsymbol\rho(\gamma) = \frac{\lfloor \rho V \rfloor}{V} 
	\quad \text{and} \quad
	\boldsymbol\varepsilon(\gamma) = \frac{\lfloor \varepsilon V \rfloor + \wp(\rho,\varepsilon,V)}{V}	
	\right\}\, 
\end{equation}
with $\wp(\rho,\varepsilon,V) = 0$ is $\lfloor \rho V \rfloor$ and $\lfloor \varepsilon V \rfloor$ have the same parity, and $\wp(\rho,\varepsilon,V) = 1$ otherwise;
later on we will often just write $\wp$ for $\wp(\rho,\varepsilon,V)$ when no confusion seems possible. 
This distinction is necessary for the set $\Gamma_{\rho,\varepsilon,V}$ not to be empty, 
since $V \boldsymbol\rho(\gamma)$ and $V \boldsymbol\varepsilon(\gamma)$ have the same parity for any $\gamma \in \Gamma_V$.
If $\Gamma_{\rho,\varepsilon,V}$ is not empty, 
the micro-canonical measure $P_{\rho,\varepsilon,V}^{\mathrm{mic}}$ is defined as the uniform measure concentrated on $\Gamma_{\rho,\varepsilon,V}$. 
The corresponding average is denoted by $\langle \cdot \rangle_{\rho,\varepsilon,V}^{\mathrm{mic}}$.

Second, given a volume $V$, an inverse temperature $\beta > 0$ and some parameter $\nu\in \R$, or $\beta = 0$ and $\nu > 0$,
the grand-canonical measure $P_{\nu,\beta,V}^{\mathrm{grd}}$ is defined as the probability measure on $\Gamma_V$ with weight
$$
	P_{\nu,\beta,V}^{\mathrm{grd}} (\gamma) 
	\; = \;
	\frac1{Z(\nu,\beta,V)} 
	\ed^{- V( \beta \boldsymbol\varepsilon(\gamma) + \nu \boldsymbol\rho (\gamma))},
	%\;\sim\;\ed^{- \beta \sum_{x \in \Lambda} (n_x - \mu/2)^2}.
$$
where $Z(\nu,\beta,V)$ ensures the normalization (partition function). 
We denote the corresponding average by $\langle \cdot \rangle_{\nu,\beta,V}^{\mathrm{grd}}$.
The grand-canonical measure is a product state and we denote the marginal on a single site by $p_{\nu,\beta}$, i.e.\@ the measure on $\N$ with density
$$
	p_{\nu,\beta} (n) \; = \; \frac{1}{z(\nu,\beta)} \ed^{- \beta n^2 - \nu n}. 
$$
For an observable $\varphi$ that depends on a fixed finite number of coordinates of $\gamma$, we will also use the notation $\langle \varphi \rangle_{\nu,\beta}$
for $\langle \varphi \rangle_{\nu,\beta,V}^{\mathrm{grd}}$. 

Let us finally define the two sets
$$
	R_{\mathrm{pos}} \; = \; 
	\{ (\rho,\varepsilon) \in ]0,\infty[^2  \; : \; 
	\varepsilon_{\mathrm{gs}} (\rho) \; < \; \varepsilon \; < \; \varepsilon_{\mathrm{c}} (\rho)\} \, ,
	\qquad 
	R_{\mathrm{inf}} \; = \; 
	\{ (\rho,\varepsilon) \in ]0,\infty[^2 \; : \; \varepsilon \; \ge \;  \varepsilon_{\mathrm{c}} (\rho) \} 
$$
with $\varepsilon_\mathrm{gs}$ and $\varepsilon_{\mathrm{c}}$ being defined by \eqref{eq: ground state energy density} and \eqref{eq: critical density} respectively. 
For $(\rho, \varepsilon) \in R_{\mathrm{pos}}$, the corresponding Gibbs state has a positive temperature ($\beta > 0$), 
while for $(\rho, \varepsilon) \in R_{\mathrm{inf}}$, the Gibbs state has infinite temperature ($\beta = 0$),   
see figure \ref{fig: phase diagram}.

Our first result deals with the equivalence of ensembles. 
Given two probability measures $P_1$ and $P_2$ on $\Gamma_V$, such that $P_2$ does not vanish on $\Gamma_V$, we define their relative entropy by 
$$
	S(P_1,P_2) \; = \; \sum_{\gamma \in \Gamma_V} P_1(\gamma) \ln \frac{P_1 (\gamma)}{P_2(\gamma)} \; \ge \; 0
$$
with the convention $0 \ln 0 = 0$. 

\begin{Theorem}\label{th: equivalence of ensembles relative entropy}
	First, given $(\rho,\varepsilon) \in R_{\mathrm{pos}}$, there exists a unique $(\nu,\beta) \in \R \times ]0,+\infty[$ such that
	$$
		S( P_{\rho,\varepsilon,V}^{\mathrm{mic}} , P_{\nu,\beta,V}^{\mathrm{grd}} ) \; = \; \mathcal O ( ( \ln V )^3 )  \qquad \text{as} \qquad V \to \infty
	$$  
	and such that the relations 
	$\langle \gamma(0) \rangle_{\nu,\beta} = \rho$ 
	and 
	$\langle (\gamma (0))^2 \rangle_{\nu,\beta} = \varepsilon$ are satisfied.
	Second, given $(\rho,\varepsilon) \in R_{\mathrm{inf}}$, there exists a unique $\nu > 0$ such that
	$$
		S( P_{\rho,\varepsilon,V}^{\mathrm{mic}} , P_{\nu,0,V}^{\mathrm{grd}} ) \; = \; \mathcal O( V^{1/2} ) \qquad \text{as} \qquad V \to \infty
	$$  
	and such that the relation
	$\langle \gamma (0) \rangle_{\nu,0} = \rho$ is satisfied.
\end{Theorem}

\begin{Remark}\label{rem: convergence of the microcanonical entropy}
As can be seen from the proof of Theorem~\ref{th: equivalence of ensembles relative entropy} 
(see (\ref{eq: proof th 1: first eq}) below),
the convergence $S( P_{\rho,\varepsilon,V}^{\mathrm{mic}} , P_{\nu,\beta,V}^{\mathrm{grd}} ) / V \to 0$ stated in Theorem~\ref{th: equivalence of ensembles relative entropy}
is equivalent to the expected scaling behavior of the microcanonical entropy:
$$
	\frac{1}{V} \ln |\Gamma_{\rho,\varepsilon,V}| \; \to \; s (p_{\nu,\beta}) 
	\qquad \text{as} \qquad 
	V \to \infty \, ,
$$ 
where $s$ is the usual entropy of a measure, see (\ref{eq: entropy}) below.
In particular, if $(\rho,\varepsilon) \in R_{\mathrm{inf}}$, 
Theorem~\ref{th: equivalence of ensembles relative entropy} 
shows that the microcanonical entropy density converges to $s(p_{\nu,0})$, 
and becomes thus independent of $\varepsilon$ in the limit $V \to \infty$, as long as $(\rho,\varepsilon) \in R_{\mathrm{inf}}$. 
\end{Remark}

As noticed for example in \cite{grosskinsky_et_al_2003}, the sub-linear growth of the relative entropy with $V$ implies that
the expectation of bounded local observables with respect to the micro-canonical measure converges to their grand-canonical expectation.
Let $\mathcal B(\Gamma_V)$ be the set of bounded functions on $\Gamma_V$. 
Moreover, by a slight abuse of notations, given $1 \le V_0 \le V$,
we consider $\mathcal B(\Gamma_{V_0})$ as the subset of functions in $\mathcal B (\Gamma_V)$ that depend only on $\gamma$ through the $V_0$ first variables.

\begin{Corollary}\label{co: equivalence of ensembles local observables}
	Let $(\rho,\varepsilon) \in R_{\mathrm{pos}} \cup R_{\mathrm{inf}}$ and let $(\nu,\beta)$ as given by Theorem \ref{th: equivalence of ensembles relative entropy}.
	Let $V_0 \ge 1$. Given $\varphi \in B (\Gamma_{V_0})$, it holds that 
	$$
		\lim_{V \to \infty} \langle \varphi \rangle_{\rho,\varepsilon,V}^{\mathrm{mic}} 
		\; = \;  
		\langle \varphi \rangle_{\nu,\beta} .
	$$
\end{Corollary} 

Let $\partial R_{\mathrm{inf}} = \{ (\rho,\varepsilon) \in ]0,+\infty[^2 : \varepsilon = \varepsilon_{\mathrm c}(\rho)\}$ 
and let us now assume that $(\rho,\varepsilon) \in R_{\mathrm{inf}} \backslash \partial R_{\mathrm{inf}}$. 
Let us define the excess of energy 
\begin{equation}\label{eq: energy excess}
	\Delta \varepsilon 
	\; = \;
	\varepsilon - \varepsilon_{\mathrm c} (\rho)
	\; = \;
	\varepsilon - \langle (\gamma(0))^2 \rangle_{\nu,0} \; > \;  0
\end{equation}
where $\nu > 0$ is the unique parameter such that $\rho  = \langle \gamma (0) \rangle_{\nu,0}$.  
Our second result asserts that, under the microcanonical measure, there is a single site containing a macroscopic quantity of energy, 
and that this quantity is equal to $\Delta \varepsilon$ up to sub-extensive corrections. 
The formulation is directly inspired by the results in \cite{chatterjee_2017}. 
Given $\gamma \in \Gamma_V$, let $\gamma_{max} = \max_{x\in \Lambda} \{ \gamma (x)\}$ and 
$\gamma_{max}'$ be the second largest coordinate of $\gamma$, 
i.e.\@ $\gamma_{max}' = \max_{x \in \Lambda\backslash x_0} \{\gamma (x)\}$ where $x_0$ is such that $\gamma(x_0) = \gamma_{max}$. 

\begin{Theorem}\label{th:single site concentration}
If $(\rho,\varepsilon) \in  R_{\mathrm{inf}}\backslash \partial R_{\mathrm{inf}}$, then
$$
	\frac{\gamma_{max}^2}{V} \quad \to \quad \Delta \varepsilon \qquad\text{and}\qquad \frac{(\gamma_{max}')^2}{V} \quad \to \quad 0
$$
as $V\to \infty$, in probability with respect to $P_{\rho,\varepsilon,V}^{\mathrm{mic}}$. 
\end{Theorem}

\section{Proof of Theorem \ref{th: equivalence of ensembles relative entropy} and Corollary \ref{co: equivalence of ensembles local observables}}
\label{sec: equivalence of ensembles}

Let us first describe the full phase diagram. Let $(\rho,\varepsilon) \in ]0,+\infty[^2$. 

\begin{Lemma}\label{lem: positive temperature phase}
$(\rho,\varepsilon) \in R_{\mathrm{pos}}$ if and only if there exists a unique $\nu \in \R$ and $\beta > 0$ so that 
\begin{equation}\label{eq: two expectations}
	\rho \; = \;  \langle \gamma (0) \rangle_{\nu,\beta}
	\qquad \text{and} \qquad 
	\varepsilon \; = \;  \langle (\gamma(0))^2 \rangle_{\nu,\beta}.
\end{equation}
\end{Lemma}

\begin{proof}
Let
$$
	g :\; \R \times ]0,+\infty[ \;\to\; \R,\; (\nu,\beta) \;\mapsto\; -\ln z(\nu,\beta) \; = \; -\ln \sum_{n\ge 0} \ed^{- \beta n^2 - \nu n} .
$$
For any $(\nu,\beta) \in \R \times ]0,+\infty[$,
$$
	\nabla g (\nu,\beta) \; = \; (\langle \gamma (0) \rangle_{\nu,\beta}, \langle (\gamma(0))^2 \rangle_{\nu,\beta}) \, ,
$$ 
and we thus need to show that $\nabla g$ is a bijection from $\R \times ]0,+\infty[$ to $R_{\mathrm{pos}}$. 
Since the Hessian of $g$ equals the opposite of the covariance matrix of $\gamma(0)$ and $(\gamma(0))^2$, 
the function $g$ is strictly concave, and $\nabla g$ is thus injective. 

Let us show that each iso-thermal line of $\nabla g$, i.e.\@ the set of values of $\nabla g$ obtained by varying $\nu$ at a fixed $\beta$, 
is the graph of a function in the $(\rho,\varepsilon)$-plane with domain $]0,+\infty[$.
This follows from the fact that $\partial_\nu (\partial_\nu g) (\nu , \beta) < 0$ for all $(\nu,\beta) \in \R \times ]0,+\infty[$ and that
$$
	\partial_\nu g (\nu,\beta) \to 0 \quad\text{as}\quad \nu \to + \infty
	\qquad\text{and}\quad
	\partial_\nu g (\nu, \beta) \to + \infty \quad\text{as}\quad \nu \to - \infty
$$	
for any given $\beta > 0$. 
Indeed, this allows to define the function $\boldsymbol \nu$ on $]0,+\infty[^2$ such that 
$$\partial_\nu g (\boldsymbol \nu (\rho,\beta),\beta) = \rho$$ for all $(\rho,\beta) \in ]0,+\infty[^2$. 
The function $f(\cdot, \beta)$ describing the isothermal line at a given $\beta$ is then 
$$
	f : ]0,+\infty[^2 \; \to  \; \R , \; (\rho,\beta) \; \to \; \partial_\beta g (\boldsymbol\nu(\rho,\beta),\beta).
$$
We define also the functions $f(\cdot, 0)$ and $f (\cdot, \infty)$, describing respectively the infinite temperature line and the ground state line, by 
$$
	f(\rho, 0) \; = \; 2 \rho^2 + \rho, 
	\qquad
	f(\rho,\infty) \; = \;  (1-(\rho - \lfloor \rho \rfloor)) \lfloor \rho \rfloor^2 + (\rho - \lfloor \rho \rfloor) (\lfloor \rho \rfloor + 1)^2
$$
for $\rho \in ]0,+\infty[$.
We notice also that the functions $\rho \mapsto f(\rho, \beta)$ are strictly increasing for all $\beta \in ]0,+\infty[$, 
as follows from the fact that 
$$
	\partial_\nu (\partial_\beta g) (\nu,\beta)
	\; = \; 
	- \langle X^3 \rangle + \langle X^2 \rangle \langle X \rangle
	\; = \; 
	- \big\langle (X + \langle X \rangle ) (X - \langle X \rangle )^2\big\rangle 
	\; < \; 0
$$
with $X = \gamma(0) \ge 0$ and $\langle \cdot \rangle = \langle \cdot \rangle_{\nu,\beta}$. 

To conclude, it is now enough to show that for any $\rho \in ]0,+\infty[$,
$$
	f(\rho, \beta) \; \to \;  f (\rho,0) \quad\text{as}\quad \beta \; \to \; 0
	\qquad \text{and} \qquad
	f(\rho, \beta) \; \to \;  f (\rho,\infty) \quad\text{as}\quad \beta \; \to \; \infty \, .
$$
Indeed, for any $\rho \in ]0,+\infty[$, the map $\beta \mapsto f(\rho,\beta)$ is continuous and monotonic, since $\nabla g$ is injective.
Hence every point $(\rho,\varepsilon)$ with $\varepsilon \in ]f(\rho,\infty),f(\rho,0)[$ will be in the image of $\nabla g$, 
while no point of the form $(\rho,\varepsilon)$ with $\varepsilon \notin ]f(\rho,\infty),f(\rho,0)[$ will be. 

Let us first consider the limit $\beta \to 0$.  
Let us fix $\rho>0$, and let $\nu>0$ be such that $\langle \gamma (0) \rangle_{\nu,0} = \rho$.
One computes that $|\partial_\nu g (\nu,\beta) - \rho| = \mathcal O(\beta)$ as $\beta \to 0$ and,
since the derivative $\partial_\nu (\partial_\nu g)$ does not vanish in the neighborhood of $\nu$ as $\beta \to 0$, 
there exists $\nu'(\beta)$ such that $|\nu - \nu'(\beta)| = \mathcal O (\beta)$ and $\partial_\nu g(\nu'(\beta),\beta) = \rho$, for any $\beta$ in a neighborhood of $0$.
Finally, one computes that $|\partial_\beta g (\nu',\beta) - \langle (\gamma(0))^2 \rangle_{\nu,0}|  = \mathcal O (\beta)$, which is the claim. 

Let us next consider the limit $\beta \to \infty$. 
Let us fix $\rho > 0$ and let us first assume $\rho \notin \N$.
Let us write 
$$
	\beta n^2 + \nu n \; = \; \beta (n - \mu)^2 - \nu^2 / 4 \beta 
	\qquad \text{with} \qquad 
	\mu \; = \; - \nu / 2 \beta,
$$
and let $\nu$ be such that $\lfloor \rho \rfloor \le \mu \le \lfloor \rho \rfloor + 1$. 
As $\beta \to \infty$, the measure $p_{\nu,\beta} = p_{-2\beta \mu,\beta}$ is very well approximated by the probability measure
$$
	p_{\mu,\beta}'
	\; = \; 
	\frac{\ed^{- \beta (\lfloor \rho \rfloor - \mu)^2}}{z'(\mu,\beta)} \delta_{\lfloor \rho \rfloor}
	\; + \; 
	\frac{\ed^{- \beta (\lfloor \rho \rfloor +1 - \mu)^2}}{z'(\mu,\beta)} \delta_{\lfloor \rho \rfloor + 1}
	\; = :\;
	(1 - \alpha(\mu,\beta)) \delta_{\lfloor \rho \rfloor} + \alpha(\mu,\beta) \delta_{\lfloor \rho \rfloor + 1}
$$
where $z'(\mu,\beta)$ is a normalization factor. 
Let $\langle \cdot \rangle_{\mu,\beta}'$ be the expectation with respect to the measure $p_{\mu,\beta}'$.
A computation shows that 
$$
	 | \langle (\gamma(0))^q \rangle_{\nu,\beta} - \langle (\gamma(0))^q \rangle_{\mu,\beta}' | \; = \; \mathcal O (\ed^{-\beta}), 
	 \qquad
	 q = 1,2, 
$$
and 
$$
	\alpha(\lfloor \rho \rfloor,\beta) \; = \; \mathcal O (\ed^{-\beta}), 
	\qquad 
	\alpha(\lfloor \rho \rfloor +1,\beta) \; = \; 1 + \mathcal O(\ed^{-\beta}).
$$
Expressing that $\langle \gamma (0) \rangle_{\nu,\beta} = \rho$ yields the relation
$$
	\rho \; = \; (1 - \alpha(\mu,\beta)) \lfloor \rho \rfloor + \alpha(\mu,\beta) (\lfloor \rho \rfloor +1) + \mathcal O (\ed^{-\beta}),
$$
hence $\alpha (\mu,\beta) = \rho - \lfloor \rho \rfloor + \mathcal O (\ed^{- \beta})$, and therefore
$$
	f(\rho,\beta) 
	\; = \; 
	(1 - \alpha(\mu,\beta)) (\lfloor \rho \rfloor)^2 + \alpha(\mu,\beta) (\lfloor \rho \rfloor +1)^2 + \mathcal O (\ed^{-\beta})
	%\; = \; 
	%(1-p) \lfloor \rho \rfloor^2 + p \lceil \rho \rceil^2 + \mathcal O (\ed^{-\beta})
	\; = \; 
	f(\rho, \infty) + \mathcal O (\ed^{-\beta})
$$
as $\beta \to \infty$.
Finally the case $\rho \in \N$ follows from the cases $\rho \notin \N$ and the fact that the map $\rho \mapsto f(\rho, \beta)$ is increasing. 
\end{proof}

To any configuration $\gamma \in \Gamma_V$, we associate a probability distribution $p(\gamma, \cdot)$ on $\N$ by 
\begin{equation}\label{eq: density of a configuration}
	p(\gamma,n) = \frac{|\{ x \in \Lambda : \gamma (x) = n \}|}{V}.
\end{equation}
Thus $p(\gamma,n)$ measures the frequency at which a site is occupied by $n$ particles, if the system is in the configuration $\gamma$. 
Let $\mathcal P_V$ be the set of probability measures on $\N$ obtained in this way. 
Equivalently, $\mathcal P_V$ is the set of probability measures $p$ on $\N$ such that $Vp(n)$ is an integer for any $n \in \N$. 
Moreover, given $(\rho,\varepsilon) \in ]0,+\infty[^2$, we denote by $\mathcal P_{\rho,\varepsilon,V}$ the subset of $\mathcal P_V$
corresponding to configurations in $\Gamma_{\rho,\varepsilon,V}$. 
Equivalently, $\mathcal P_{\rho,\varepsilon,V}$ is the set of measures $p\in \mathcal P_V$ such that 
\begin{equation}\label{eq: density energy constraints for measures}
	\sum_{n\ge 0} n p(n) \; = \; \frac{\lfloor \rho V \rfloor}{V}, 
	\qquad 
	\sum_{n\ge 0} n^2 p(n) \; = \; \frac{\lfloor \varepsilon V \rfloor + \wp}{V}
\end{equation}

The entropy of a probability measure $p$ on $\N$ is defined as 
\begin{equation}\label{eq: entropy}
	s(p) \; = \; - \sum_{n\ge 0} p(n) \ln p(n)
\end{equation}
with the convention $0 \ln 0 = 0$.
If $p \in \mathcal P_V$, the number of configurations in $\Gamma_V$ giving rise to the measure $p$ can be computed explicitly: 
$$
	|\{ \gamma \in \Gamma_V : p(\gamma,\cdot) = p \}| \; = \; \frac{V !}{(p(0)V)! \dots (p(V)V)!} 
$$
and, using rigorous bounds on the Stirling's approximation, 
$$
	(2\pi n)^{1/2} (n/\ed)^n \; \le \; n ! \; \le \; (\ed^2 n)^{1/2} (n/\ed)^n, 
$$
valid for any $n\ge 1$, we get 
\begin{equation}\label{eq: bound from Stirling}
	\left(\prod_{n\in\N:p(n) > 0} (\ed^2 p(n)V)^{-1/2} \right) \ed^{s(p) V}
	\; \le \; 
	|\{ \gamma \in \Gamma_V : p(\gamma,\cdot) = p \}|
	\; \le \; 
	(\ed^2 V)^{1/2} \;  \ed^{s(p) V}.
\end{equation}
Finally, given a probability measure $p$ on $\N$, we define $\mathrm{supp} (p) = \{ n \in \N : p(n) > 0 \}$.

The key input for our proof of equivalence of ensembles is to find a state 
$p_V \in \mathcal P_{\rho,\varepsilon,V}$ that converges to some Gibbs state $p_{\nu,\beta}$ as $V \to \infty$, 
and we actually only need to know that the entropy of $p_V$ converges to the entropy of this Gibbs state. 
If $(\rho,\varepsilon)\in R_{\mathrm{pos}}$, 
the natural (and actually only possible) choice of thermodynamical parameters is furnished by Lemma~\ref{lem: positive temperature phase},
and it is possible to construct the state $p_V$ through some very explicit approximations.  
If $(\rho,\varepsilon)\in R_{\mathrm{inf}}$, we may instead consider an infinite temperature Gibbs state at the right density of particles. 
In the physical system, the excess of energy $\Delta \varepsilon$ may be concentrated on a single site by accumulating only about $V^{1/2}$ particles on that site. 
Implementing this idea to generate $p_V$ will produce an entropy difference that decays as $V^{-1/2}$. 
The strategy of concentrating the whole excess of energy was used in \cite{grosskinsky_et_al_2003}, 
and is a priori independent on the fact that it does indeed happen with an overwhelming probability, as Theorem~\ref{th:single site concentration} shows. 

%Given a measure $p \in \Gamma_V$, let us define 
%$\mathrm M(p) = \max\{n\in\N : p(n) > 0\}$ 
%and $\mathrm M'(p) = \max\{n < \mathrm M(p): p(n) > 0 \}$ with the convention $\max \varnothing = 0$.

\begin{Lemma}\label{lem: discrete approximation Gibbs state}
	Let first $(\rho,\varepsilon)\in R_{\mathrm{pos}} \cup \partial R_{\mathrm{inf}}$, 
	and let $(\nu,\beta)$ be the corresponding thermodynamical parameters furnished by Lemma~\ref{lem: positive temperature phase}. 
	For any $V\ge 1$, there exists $p_V \in \mathcal P_{\rho,\varepsilon,V}$ so that 
	\begin{equation}\label{eq: entropy difference 1}
		|s(p_V) - s(p_{\nu,\beta})| \; = \; \mathcal O \left(\frac{(\ln V)^3}{V}\right)  % \frac{(\ln V)^3}{V}.  
	\end{equation}
	and $|\mathrm{supp}(p_V)| = \mathcal O(\ln V)$ as $V \to \infty$.
	Let then $(\rho,\varepsilon)\in R_{\mathrm{inf}} \backslash \partial R_{\mathrm{inf}}$, 
	and let $\nu$ be such that $\rho = \langle \gamma (0) \rangle_{\nu ,0}$. 
	For any $V \ge 1$, there exists $p_V \in \mathcal P_{\rho,\varepsilon,V}$ so that 
	\begin{equation}\label{eq: entropy difference 2}
		|s(p_V) - s(p_{\nu,0})| \; = \; \mathcal O (V^{-1/2})
	\end{equation}
	and $|\mathrm{supp}(p_V)| = \mathcal O (\ln V)$ as $V \to \infty$.
\end{Lemma}

\begin{proof}
We only need to consider $V \ge V_0$, for some given $V_0$. 

Let us first deal with the case $(\rho,\varepsilon)\in R_{\mathrm{pos}}\cup \partial R_{\mathrm{inf}}$, and let us construct the measure $p_V$ in three steps. 
Let us first define the probability measure $p^{(0)}$ on $\N$ by 
$$
	p^{(0)} (n) \; = \;  \frac{\lfloor V p_{\nu,\beta} (n) \rfloor}{V} \quad\text{for}\quad n\ge 1
$$
and $p^{(0)}(0) = 1 - \sum_{n\ge 1} p^{(0)} (n) \ge 0$ %$ \ge p_{\nu,\beta} (0)$ 
so that $p^{(0)} \in \mathcal P_V$.
In order to satisfy the constraint on the density of particles, we define $p^{(1)}$ by 
$p^{(1)}(n) = p^{(0)}(n)$ for $n \ge 2$ and 
\begin{equation}\label{eq: p 1 measure}
	p^{(1)}(0) \; = \;  p^{(0)} (0) - \delta_1, \quad 
	p^{(1)}(1) \; = \;  p^{(0)} (1) +  \delta_1
	\qquad  \text{with} \qquad
	\delta_1 \; = \; \frac{\lfloor \rho V \rfloor}{V} - \sum_{n \ge 0} n p^{(0)} (n) \, .
\end{equation}
%(we will check later on that $p^{(1)}(0)\ge 0$ and $p^{(1)}(1) \ge 0$ if $V_0$ is taken large enough).
Finally, in order to satisfy the constraint on the density of energy, we define $p^{(2)}$ by 
$p^{(2)}(n) = p^{(1)}(n)$ for $n \ge 3$ and 
\begin{equation}\label{eq: p 2 measure}
	\begin{split}
	&p^{(2)}(0) \; = \; p^{(1)}(0) + \delta_2, \quad
	p^{(2)}(1) \; = \; p^{(1)}(1) - 2\delta_2, \quad
	p^{(2)}(2) \; = \; p^{(1)}(2) + \delta_2, \\ 
	&\delta_2 \; = \; \frac{1}{2} 
	\left( \frac{\lfloor \varepsilon V \rfloor + \wp}{V} - \sum_{n\ge 0}n^2 p^{(1)}(n) \right) \, .
	\end{split}
\end{equation}
%(again, we will check later that $p^{(2)}(0) \ge 0$, $p^{(2)}(1)\ge 0$ and $p^{(2)}(2) \ge 0$ for $V_0$ large enough). 
Notice that $V\sum_{n\ge 0}n^2 p^{(1)}(n)$ has the same parity as $\lfloor \rho V \rfloor$, 
hence the same parity as $\lfloor \varepsilon V \rfloor + \wp$, so that $V \delta_2 \in \N$. 
Elementary bounds yield
$$
	|\delta_1| \; = \; \mathcal O \left( \frac{(\ln V)^2}{V} \right) 
	\qquad \text{and} \qquad
	|\delta_2| \; = \; \mathcal O \left( \frac{(\ln V)^3}{V} \right)
$$
as $V \to \infty$. 
In particular, for $V_0$ large enough, $p^{(2)}(0)\ge 0$, $p^{(2)}(1) \ge 0$ and $p^{(2)}(2) \ge 0$, and the above construction implies that 
$p^{(2)}\in \mathcal P_{\rho,\varepsilon,V}$.
Moreover, $|\mathrm{supp}(p^{(2)})| = |\mathrm{supp}(p^{(0)})| = \mathcal O( \ln V) $, 
and a computation yields the bound \eqref{eq: entropy difference 1} with $p^{(2)}$ instead of $p_V$. 
We set finally $p_V = p^{(2)}$. 

We deal with the case where $(\rho,\varepsilon) \in R_{\mathrm{inf}}\backslash\partial R_{\mathrm{inf}}$ through a similar strategy.
Let $\Delta \varepsilon$ be the energy excess defined in \eqref{eq: energy excess} and let us concentrate this excess on a single site: 
Let us this time define $p^{(0)}$ by 
$$
	p^{(0)} (n) \; = \;  \frac{\lfloor V p_{\nu,0} (n) \rfloor}{V} \quad\text{for}\quad n\ge 1, \; n \ne  \lfloor (V\Delta \varepsilon)^{1/2}\rfloor, 
	\qquad 
	p^{(0)} (\lfloor (V\Delta \varepsilon)^{1/2}\rfloor) \; = \; \frac{1}{V}
$$
and $p^{(0)}(0) = 1 - \sum_{n\ge 1} p^{(0)} (n) \ge p_{\nu,0}(0) - 1/V \ge 0$, so that $p^{(0)} \in \mathcal P_V$, for $V_0$ large enough. 
The measures $p^{(1)}$ and $p^{(2)}$ are defined as before through \eqref{eq: p 1 measure} and \eqref{eq: p 2 measure} respectively. 
This time however, we find only
$$
	|\delta_1| \; = \; \mathcal O(V^{-1/2}) 
	\qquad \text{and} \qquad
	|\delta_2| \; = \; \mathcal O (V^{-1/2})
$$
as $V\to \infty$.
One concludes as in the previous case and again, one sets finally $p_V = p^{(2)}$.
\end{proof}

\begin{proof}[Proof of Theorem \ref{th: equivalence of ensembles relative entropy}]
Let $(\rho,\varepsilon) \in R_{\mathrm{pos}} \cup R_{\mathrm{inf}}$. 
If $(\rho,\varepsilon) \in R_{\mathrm{pos}}$, let $(\nu,\beta)$ be given by Lemma~\ref{lem: positive temperature phase}. 
Otherwise let $\nu$ be such that $\langle \gamma (0) \rangle_{\nu,0}=\rho$ and let $\beta = 0$. 
We start with 
\begin{align} 
	S( P_{\rho,\varepsilon,V}^{\mathrm{mic}} , P_{\nu,\beta,V}^{\mathrm{grd}} )
	\; &= \;
	\sum_{\gamma \in \Gamma_V} P_{\rho,\varepsilon,V}^{\mathrm{mic}}(\gamma) 
	\ln \frac{P_{\rho,\varepsilon,V}^{\mathrm{mic}}(\gamma)}{P_{\nu,\beta,V}^{\mathrm{grd}}(\gamma)}
	\; = \; 
	\ln \frac{Z(\nu,\beta,V)}{|\Gamma_{\rho,\varepsilon,V}| \ed^{- \nu \lfloor \rho V \rfloor - \beta (\lfloor \varepsilon V \rfloor + \wp)}} 
	\nonumber\\
	\; &\le \;
	- \ln |\Gamma_{\rho,\varepsilon,V}|
	- \ln \frac{\ed^{-(\nu \rho + \beta\varepsilon)V}}{(z(\nu,\beta))^V} + c
	\; = \; 
	s(p_{\nu,\beta}) V - \ln |\Gamma_{\rho,\varepsilon,V}| + c,
	\label{eq: proof th 1: first eq}
\end{align}
for some constant $c < + \infty$. 
To get a lower bound on $\ln |\Gamma_{\rho,\varepsilon,V}|$, we will use the first inequality in \eqref{eq: bound from Stirling}
with $p = p_V$ as given by by Lemma~\ref{lem: discrete approximation Gibbs state}.
Since $|\mathrm{supp}(p_V)| = \mathcal O (\ln V)$, 
the product in the left-hand side of \eqref{eq: bound from Stirling} is lower bounded by $V^{-c \ln V}$ for some $c < +\infty$, and thus 
\begin{equation}\label{eq: last in proof th 1}
	\ln |\Gamma_{\rho,\varepsilon,V}| 
	\; \ge \;
	- c (\ln V)^2 + s(p_V) V.
\end{equation}
Finally, we use the estimate \eqref{eq: entropy difference 1} if $(\rho,\varepsilon) \in R_{\mathrm{pos}} \cup \partial R_{\mathrm{inf}}$,
and \eqref{eq: entropy difference 2} if $(\rho,\varepsilon) \in R_{\mathrm{inf}} \backslash \partial R_{\mathrm{inf}}$, 
in order to replace $s(p_V)$ by $s(p_{\nu,\beta})$ in \eqref{eq: last in proof th 1} up to a small error, 
and we obtain the result by inserting \eqref{eq: last in proof th 1} in \eqref{eq: proof th 1: first eq}.   
\end{proof}

\begin{proof}[Proof of Corollary \ref{co: equivalence of ensembles local observables}.]
We follow \cite{grosskinsky_et_al_2003}.
The marginal of $P^{\mathrm{grd}}_{\nu,\beta,V}$ on $\Gamma_{V_0}$ is simply $P^{\mathrm{grd}}_{\nu,\beta,V_0}$, and 
we denote the marginal of $P^{\mathrm{mic}}_{\rho,e,V}$ by $P^{\mathrm{mic},V_0}_{\rho,e,V}$. 
By subadditivity of the relative entropy, 
$$
	S(P^{\mathrm{mic},V_0}_{\rho,e,V},P^{\mathrm{grd}}_{\nu,\beta,V_0}) 
	\; \le \;
	\left\lfloor\frac{V_0}{V}\right\rfloor S(P^{\mathrm{mic}}_{\rho,e,V},P^{\mathrm{grd}}_{\nu,\beta,V}).
$$  
By Theorem \ref{th: equivalence of ensembles relative entropy}, the right hand side goes to $0$ as $V \to \infty$. 
This implies the weak convergence $P^{\mathrm{mic},V_0}_{\rho,e,V} \to P^{\mathrm{grd}}_{\nu,\beta,V_0}$, hence our claim. 
\end{proof}

\section{Proof of Theorem~\ref{th:single site concentration}}\label{sec: concentration on a single site}

As stressed in the introduction, while the entropy density stays constant in the thermodynamic limit as $\varepsilon$ increases for $\varepsilon > \varepsilon_c (\rho)$, 
the total entropy decreases sub-extensively.
Assuming that the system will indeed form a condensate with about $(V\Delta \varepsilon)^{1/2}$ particles, 
the number of particles in the bulk decays by the same amount, hence also the entropy gets reduced by a amount proportional to $(V\Delta \varepsilon)^{1/2}$. 
As an upshot, it is good to have in mind that the entropy differences that are relevant in this part are of order $V^{1/2}$ at least.   
Our first lemma yields a bound on the entropy resulting from the number of measures in $\mathcal P_{\rho,\varepsilon,V}$ and shows that this entropy is negligible.  

\begin{Lemma}\label{lem: number of states}
For $(\rho,\varepsilon) \in ]0,+\infty [^2$, 
$
	\ln |\mathcal P_{\rho,\varepsilon,V}| \; = \; \mathcal O (V^{1/3} \ln V)
$
as $V \to \infty$. 
\end{Lemma}

\begin{proof}
%%%
%%%
% formula : 
% (\lfloor N^{1/2}\rfloor + 1)^2 \ge N
%%%
%%%
Let $N\in \N$ be large enough for all expressions below to make sense and let 
$$
	Z_N
	\; = \; 
	\left|\left\{ 
	\mathrm k \in \mathcal \N^{\lfloor N^{1/2}\rfloor} : \sum_{n = 1}^{\lfloor N^{1/2}\rfloor} n^2 \mathrm k(n) \le N 
	\right\}\right| \, .
$$
For any $r \in \{ 1, \dots , \lfloor N^{1/2}\rfloor - 1\}$, it holds that 
$$
	Z_N
	\; \le \; 
	N^{r}
	\left|\left\{
	(\mathrm k (r+1), \dots , \mathrm k(\lfloor N^{1/2}\rfloor)) \in \N^{\lfloor N^{1/2}\rfloor-r} : \sum_{n=r+1}^{\lfloor N^{1/2}\rfloor} n^2 \mathrm k(n) \le N
	\right\}\right| \, .
$$
Let now $r = \lfloor N^\alpha \rfloor$ for some $1/4 < \alpha < 1/2$ to be fixed later. 
Since $n^2 \ge (\lfloor N^\alpha \rfloor + 1)^2$ in the last sum, 
we have the bound $\mathrm k (n) \le N^{1 - 2 \alpha}$ 
and we know that $\mathrm k \in \N^{\lfloor N^{1/2}\rfloor-\lfloor N^\alpha \rfloor}$ has at most $\lfloor N^{1 - 2 \alpha}\rfloor $ non-zero coordinates. 
Hence
$$
	Z_N 
	\; \le \; 
	N^{N^\alpha}
	\left(
	\begin{array}{c}
	\lfloor N^{1/2}\rfloor - \lfloor N^\alpha \rfloor \\
	\lfloor N^{1 - 2 \alpha}\rfloor
	\end{array}
	\right)
	\left( N^{1 - 2 \alpha} \right)^{N^{1 - 2 \alpha}} 
	\; \le \; 
	\ed^{c \ln N (N^\alpha + N^{1 - 2 \alpha})}
$$
for some constant $c < + \infty$.
Taking $\alpha = 1/3$ yields the bound $\ln Z_N = \mathcal O (N^{1/3}\ln N)$ as $N\to \infty$. 
Finally, to get the claim, we observe that $|\mathcal P_{\rho,\varepsilon,V}| \le Z_{N}$ with $N = \lfloor \varepsilon V \rfloor + 1$.  
\end{proof}

Given $I \in \N\backslash \{ 0 \}$ and $(\nu,\beta)\in \R^2$, we define a probability measure on $\N$ by 
$$
	p_{\nu,\beta}^I (n) \; = \; \frac{1_{\{ n\le I \} } \ed^{- \nu n - \beta n^2}}{z^I(\nu,\beta)}
$$
where $z^I(\nu,\beta)$ ensures the normalization. 
This corresponds to the one-site marginal of a Gibbs state for a system where the number of particles per sites is imposed to be at most $I$. 
The next lemma contains the crucial input to show the condensation on a single site: 
expressions \eqref{eq: definition of rho I} and \eqref{eq: entropy difference with cut-off I} below
furnish a practical way to estimate the entropy of the Gibbs states $p_{\nu,\beta}^I$ when $(\rho,\varepsilon) \in R_{\mathrm{inf}} \backslash \partial R_{\mathrm{inf}}$. 
For an appropriate choice of $(\nu,\beta)$, 
the Gibbs state $p_{\nu,\beta}^I$ has maximal entropy among all probability measures concentrated on $\{0, \dots ,I \}$ 
and satisfying the density and energy constraints \eqref{eq: density energy constraints for measures}. 
Hence this lemma furnishes an explicit bound on the entropy of all states having at most $I$ particles per site 
and will eventually allow to conclude that $\gamma_{max}^2 / V$ is not much smaller than $\Delta \varepsilon$ with high probability. 
The proof of the lemma is based on an approximate guess for the value of $(\nu,\beta)$ 
so that the constraints constraints \eqref{eq: density energy constraints for measures} are satisfied, 
as expressed in \eqref{eq: variables delta nu delta beta} below.

\begin{Lemma}\label{lem: entropy gibbs on truncated space}
	Let $(\rho,\varepsilon) \in R_{\mathrm{inf}} \backslash \partial R_{\mathrm{inf}}$,
	let 
	\begin{equation}\label{eq: definition of rho I}
		\rho_I \; = \;  \frac{\lfloor \rho V\rfloor }{V} - \frac{\Delta\varepsilon}{I} % \; > \; 0 ,
	\end{equation}
	and let then $\nu_I$ be such that $\langle \gamma (0) \rangle_{\nu_I,0} = \rho_I$, assuming that $I,V$ are large enough so that $\rho_I > 0$. 
	%There exists $I_0 \in \N$ such that, for any $I \in \N \backslash \{0 \}$ with $I\ge I_0$, 
	For any such $I,V$,
	there exists a unique $(\nu,\beta)\in \R^2$ such that $\langle \gamma (0) \rangle_{\nu,\beta}^I = \lfloor \rho V \rfloor / V$ 
	and 
	$\langle (\gamma(0))^2 \rangle_{\nu,\beta}^I = (\lfloor \varepsilon V \rfloor + \wp )/V$, 
	and moreover
	\begin{equation}\label{eq: entropy difference with cut-off I}
		|s(p^I_{\nu,\beta}) - s(p_{\nu_I,0})|  \; \le \;  \mathcal O \left(\frac{\ln I}{I^2}\right)
	\end{equation}
	as $I \to \infty$. 
\end{Lemma}

\begin{Remark}
The value of $(\nu,\beta)$ in the above lemma depends obviously on $I$ and $V$, and the value of $\rho_I$ and $\nu_I$ depends on $V$. 
\end{Remark}

\begin{Remark}
Later on, in the proof of Theorem~\ref{th:single site concentration}, 
the cut-off $I$ will be taken (slightly smaller than) $(V \Delta \varepsilon)^{1/2}$, 
hence the density $\rho_I$ will be (slightly smaller than) the density resulting from the condensation of all the excess of energy on a single site. 
\end{Remark}

\begin{proof}
The dependence on $V$ is completely irrelevant and to simplify our expressions, 
we will simply write $\rho$ for $\lfloor \rho V \rfloor / V$ and $\varepsilon$ for $(\lfloor \varepsilon V \rfloor + \wp )/V$. 
The mere existence and unicity of the parameters $(\nu,\beta)$ follows by maximizing the entropy defined in \eqref{eq: entropy}
under the constraints \eqref{eq: density energy constraints for measures}.
But to get an estimate on the entropy $s(p^I_{\nu,\beta})$, we need to know how $(\nu,\beta)$ behaves asymptotically as $I \to \infty$. 

Let $\nu_0 > 0$ be such that $\rho = \langle \gamma(0) \rangle_{\nu_0,0}$, and let us introduce the variables $(\delta \nu, \delta \beta) \in \R^2$ defined by
\begin{equation}\label{eq: variables delta nu delta beta}
	\nu \; = \; \nu_0 + \frac{\delta \nu}{I} 
	\qquad \text{and} \qquad 
	\beta \; = \; - \frac{1}{I} \left( \nu_0 + \frac{\delta \nu}{I} \right) + \frac{2 \ln I}{I^2} + \frac{\delta \beta}{I^2}.
\end{equation}
In terms of these variables, the partition function of the system reads 
$$
	z^I(\nu, \beta) \; = \; \sum_{n=0}^I \, 
	\ed^{- \left(\nu_0 + \frac{\delta \nu}{I}\right) n \left(1 - \frac{n}{I}\right)} \, 
	\ed^{- \delta \beta \left(\frac{n}{I}\right)^2} \, 
	\ed^{- \left(\frac{n}{I}\right)^2 \ln (I^2)}
$$
and one checks that $z^I(\nu , \beta) \to z(\nu_0,0)$ as $I \to \infty$ for given $\delta \beta$ 
and for $\delta \nu$ varying possibly with $I$, but in such a way that $\delta \nu / I \to 0$. 
Thanks to this, one verifies that 
\begin{equation}\label{eq: varying beta}
	\lim_{\delta \beta \to + \infty} \lim_{I \to \infty}  \langle (\gamma(0))^2 \rangle_{\nu,\beta}^I \; = \; \langle (\gamma (0))^2 \rangle_{\nu_0,0}
	\qquad\text{and}\qquad
	\lim_{\delta \beta \to - \infty} \lim_{I \to \infty}  \langle (\gamma(0))^2 \rangle_{\nu,\beta}^I \; = \; + \infty
\end{equation}
where again $\delta \nu$ is allowed to vary with $I$ in such a way that $\delta \nu / I \to 0$ as $I \to \infty$. 
Finally, one finds a constant $c > 0$ so that 
\begin{equation}\label{eq: varying nu}
	\frac{\partial}{\partial (\delta \nu)} \langle \gamma (0) \rangle_{\nu,\beta} \; \ge \; \frac{c}{I}
\end{equation}
for any $\delta \beta$ in a compact interval, and any $\delta \nu$ such that $|\delta \nu |/ I \le 1$. 
Therefore, by \eqref{eq: varying beta}, for any $\nu \in [- I,I]$, there exists $\delta\beta(\delta\nu)$ that remains bounded as $I\to \infty$ and such that 
$\langle (\gamma (0))^2 \rangle_{\nu,\beta}^I = \varepsilon$.
Next, by \eqref{eq: varying nu}, one finds $\delta \nu$ that remains bounded as $I \to \infty$ and so that  
$\langle \gamma (0) \rangle_{\nu,\beta(\delta \nu)}^I = \rho$ 
where $\beta(\delta \nu)$ is the value of $\beta$ obtained from \eqref{eq: variables delta nu delta beta} with $\delta \beta = \delta \beta (\delta \nu)$. 

We have thus reached the conclusion that if $(\nu,\beta)$ 
is such that $\rho = \langle \gamma (0)\rangle_{\nu,\beta}^I$ and $\varepsilon = \langle (\gamma(0))^2 \rangle_{\nu,\beta}^I$, 
then $\delta\nu$ and $\delta\beta$ remain bounded as $I\to \infty$. 
We assume from now on that $(\nu,\beta)$ is as in the statement of the lemma, 
and we move to the estimate on $s(p^I_{\nu,\beta})$: %Let $\rho'$ and $\nu'$ be as in the statement of the lemma, and let  
\begin{equation}\label{eq: some entropy difference}
	s(p^I_{\nu,\beta}) - s(p_{\nu_I,0})
	\; = \; 
	\nu \rho + \beta \varepsilon - \nu_I \rho_I + \ln z^I(\nu,\beta) - \ln z(\nu_I,0).
\end{equation}
To estimate the difference of logarithms, let us first notice that
$$
	z^I (\nu,\beta) \; = \; z^{\lfloor I/2 \rfloor}(\nu,\beta) + \mathcal O (I^{-2})
	\qquad \text{and} \qquad 
	z(\nu_0,0) \; = \; z^{\lfloor I/2 \rfloor} (\nu_I,0) + \mathcal O (e^{-c I})
$$
for some $c > 0$. 
It may be worth pointing out that the parameters $(\nu,\beta)$ are obviously the same on both side of the first of these equations, 
and that $z^{\lfloor I/2 \rfloor}$ simply represents the truncation in the sum defining the partition function after the $\lfloor I/2 \rfloor$ first terms.
This truncation allows to perform a second order expansion
since the function $(\tilde \nu, \tilde \beta) \mapsto z^{\lfloor I/2 \rfloor} (\tilde \nu,\tilde \beta)$ 
admits bounded second derivatives on the segment joining $(\nu_I,0)$ to $(\nu,\beta)$: 
\begin{align}
	\frac{z^{\lfloor I/2\rfloor } (\nu,\beta)}{z^{\lfloor I/2\rfloor}(\nu_I,0)} 
	\; &= \; 
	1 + \frac{1}{z^{\lfloor I/2 \rfloor}(\nu_I,0)} \frac{\partial z^{\lfloor I/2 \rfloor}}{\partial \nu}(\nu_I,0) (\nu - \nu_I) + 
	\frac{1}{z^{\lfloor I/2 \rfloor}(\nu_I,0)} \frac{\partial z^{\lfloor I/2 \rfloor}}{\partial \beta}(\nu_I,0) \beta 
	+ \mathcal O \left(\frac{1}{I^{2}}\right)
	\nonumber\\ 
	\; &= \;
	1 - \rho_I (\nu - \nu_I) - \varepsilon_I \beta + \mathcal O (I^{-2})
	\label{eq: bound on the quotient of z}
\end{align}
with $\varepsilon_I = \langle (\gamma (0))^2 \rangle_{\nu_I,0}$. 
To get the bound on the first line, we used that $\beta^2, (\nu - \nu_I) = \mathcal O (I^{-2})$, 
as follows from \eqref{eq: variables delta nu delta beta} and \eqref{eq: definition of rho I}, 
while to get the second line, we replaced again $z^{\lfloor I/2\rfloor}$ by $z$ up to an exponentially small error in $I$.
Inserting the bound \eqref{eq: bound on the quotient of z} in \eqref{eq: some entropy difference}, we get 
$$
	s(p^I_{\nu,\beta}) - s(p_{\nu_I,0})
	\; = \; 
	\nu (\rho - \rho_I) + \beta (\varepsilon - \varepsilon_I) + \mathcal O (I^{-2}). 
$$
For the first term we get directly from \eqref{eq: definition of rho I} that
$$
	\nu (\rho - \rho_I) \; = \; \nu_0 \frac{\Delta \varepsilon}{I} + \mathcal O (I^{-2})
$$
and for the second
\begin{align*}
	\beta (\varepsilon - \varepsilon_I) 
	\; &= \;
	\beta \Delta \varepsilon  + \beta (\langle (\gamma(0))^2\rangle_{\nu_0,0} - \langle (\gamma(0))^2\rangle_{\nu_I,0})\\
	\; &= \; 
	-\nu_0 \frac{\Delta \varepsilon}{I} - 
	\frac{\nu_0}{I} \big(\langle (\gamma(0))^2\rangle_{\nu_0,0} - \langle (\gamma(0))^2\rangle_{\nu_I,0}\big) + \mathcal O (I^{-2}\ln I)
	\; = \; 
	-\nu_0 \frac{\Delta \varepsilon}{I} + \mathcal O (I^{-2}\ln I)
\end{align*}
by \eqref{eq: variables delta nu delta beta}.
\end{proof}

Finally, we need an improvement on the bound \eqref{eq: entropy difference 2} in Lemma~\ref{lem: discrete approximation Gibbs state}
since an entropy difference of order $\mathcal O (V^{1/2})$ is no longer irrelevant. 
Looking back at the proof of Lemma~\ref{lem: discrete approximation Gibbs state}, this may be achieved by replacing 
the density $\rho$ by a density $\rho_V$ that depends on the volume, 
so as to compensate for the loss of mass due to particles concentrating on the single site where the excess of energy is realized. 

\begin{Lemma}\label{lem: high entropy state}
Let $(\rho,\varepsilon) \in R_{\mathrm{inf}} \backslash \partial R_{\mathrm{inf}}$, 
let 
$$
	\rho_V \; = \;  \rho - \frac{\lfloor (V \Delta \varepsilon)^{1/2}\rfloor}{V}  %\; > \; 0 ,
$$
and let $\nu_V$ be such that $\langle \gamma (0) \rangle_{\nu_V,0} = \rho_V$, assuming that $V$ is large enough so that $\rho_V > 0$. 
For any such $V$, 
there exists $p_V \in \mathcal P_{\rho,\varepsilon,V}$ so that 
\begin{equation}\label{eq: entropy difference improved}
	|s(p_V) - s(p_{\nu_V,0})| \; = \; \mathcal O (V^{-3/4}) 
\end{equation}
as $V \to \infty$ and $|\mathrm{supp} (p_V)| = \mathcal O (\ln V)$. 
\end{Lemma}

\begin{proof}
We follow the same lines as in the proof of Lemma \ref{lem: discrete approximation Gibbs state}. 
We only need to consider $V \ge \tilde V_0$ for some fixed $\tilde V_0$.
Let us define $p^{(0)}$ by 
\begin{align*}
	&p^{(0)} (n) \; = \;  \frac{\lfloor V p_{\nu_V,0} (n) \rfloor}{V} 
	\qquad\text{for}\qquad n\ge 1, \quad n \ne  \lfloor (V\Delta \varepsilon)^{1/2}\rfloor,
	\quad n \ne \lfloor \alpha V^{1/4} \rfloor \, , \\  %(V \Delta \varepsilon - \lfloor (V\Delta \varepsilon)^{1/2}\rfloor^2)^{1/2} \rfloor \, , \\
	&p^{(0)} (\lfloor (V\Delta \varepsilon)^{1/2}\rfloor) \; = \; \frac{1}{V} \, ,
	\qquad 
	p^{(0)} (\lfloor \alpha V^{1/4} \rfloor) \; = \; \frac1V
\end{align*}
for some $\alpha \ge 0$ to be determined later, and $p^{(0)}(0) = 1 - \sum_{n\ge 1} p^{(0)} (n)$. 
Since $p^{(0)}(0) \ge p_{\nu_V,0}(0) - 2/V$, it holds that $p^{(0)} \in \mathcal P_V$ for $\tilde V_0$ large enough. 
The measures $p^{(1)}$ and $p^{(2)}$ are defined as in the proof of Lemma \ref{lem: discrete approximation Gibbs state} 
through \eqref{eq: p 1 measure} and \eqref{eq: p 2 measure} respectively. 
To get a bound on $\delta_1$, we compute it from its definition in \eqref{eq: p 1 measure}: 
$$
	\delta_1 
	\; = \;  
	\left( \rho - \rho_V - \frac{\lfloor (V \Delta \varepsilon)^{1/2} \rfloor}{V} \right)
	\; + \;  \sum_{n\ge 1} \frac{\lfloor V p_{\nu_V,0}(n)\rfloor}{V} (n - \rho_V)
	- \frac{\lfloor \alpha V^{1/4} \rfloor}{V}
	+ \frac{\lfloor \rho V \rfloor - \rho V}{V}.
$$
The term in the parenthesis vanishes by the definition of $\rho_V$ and the rest is seen to be $\mathcal O (V^{-3/4})$ in absolute value for any fixed $\alpha \ge 0$. 
Thus $|\delta_1|  = \mathcal O (V^{-3/4})$.
Next we compute $\delta_2$ from its definition in \eqref{eq: p 2 measure}:
\begin{align*}
	2 \delta_2
	\; = &\;
	\varepsilon - \langle (\gamma(0))^2 \rangle_{\nu_V,0} - \frac{\lfloor ( V \Delta \varepsilon)^{1/2}\rfloor^2}{V} - \frac{\lfloor \alpha V^{1/4} \rfloor^2}{V}\\
	&
	\; + \;  \sum_{n\ge 1} \frac{p_{\nu_V,0}(n)}{V} (n^2 - \langle (\gamma(0))^2 \rangle_{\nu_V,0}) 
	\; + \;  \frac{\lfloor \varepsilon V \rfloor - \varepsilon V + \wp}{V}
	\; - \; \delta_1 \, .
\end{align*}
The sum of the terms in the second line is seem to be $\mathcal O (V^{-3/4})$, 
and we now want to adjust $\alpha$ so that the sum of terms in the first line is too. 
Since both $\langle (\gamma(0))^2 \rangle_{\nu,0} = 2 \rho^2 + \rho$ and $\langle (\gamma(0))^2 \rangle_{\nu_V,0} = 2 \rho_V^2 + \rho_V$, we write
\begin{align*}
	&\varepsilon - \langle (\gamma(0))^2 \rangle_{\nu_V,0} - \frac{\lfloor ( V \Delta \varepsilon)^{1/2}\rfloor^2}{V} \\
	&\phantom{.}\qquad\qquad\phantom{.} = \; 
	\langle (\gamma(0))^2 \rangle_{\nu,0} - \langle (\gamma(0))^2 \rangle_{\nu_V,0} + \frac{V \Delta \varepsilon - \lfloor ( V \Delta \varepsilon)^{1/2}\rfloor^2}{V} \\
	&\phantom{.}\qquad\qquad\phantom{.} = \; 
	2 \left(\rho^2 - \left(\rho - \frac{\lfloor ( V \Delta \varepsilon)^{1/2}\rfloor}{V}\right)^2\right) + \frac{\lfloor ( V \Delta \varepsilon)^{1/2}\rfloor}{V}
	+ \frac{V \Delta \varepsilon - \lfloor ( V \Delta \varepsilon)^{1/2}\rfloor^2}{V} \\
	&\phantom{.}\qquad\qquad\phantom{.} = : \; \frac{\theta_V}{V^{1/2}}
\end{align*}
with $0 \le \theta_V = \mathcal O (1)$. 
Therefore, taking $\alpha = \theta_V^{1/2}$, one finds that $|\delta_2| = \mathcal O (V^{-3/4})$.  
One concludes as in the proof of Lemma \ref{lem: discrete approximation Gibbs state}, and one sets finally $p_V = p^{(2)}$.
\end{proof}

Thanks to the three above lemmas, we can now come to the 

\begin{proof}[Proof of Theorem \ref{th:single site concentration}]  
The micro-canonical probability of an event $E \subset \Gamma_{\rho,\varepsilon,V}$ is given by 
$$
	P_{\rho,e,V}^{\mathrm{mic}} (E)
	\; = \; 
	\frac{|E|}{|\Gamma_{\rho,\varepsilon,V}|} .
$$
We only will need upper bounds on this probability and, proceeding as in the proof of Theorem \ref{th: equivalence of ensembles relative entropy}, 
we use Lemma \ref{lem: high entropy state} to get 
\begin{equation}\label{eq: proof end 1}
	\ln |\Gamma_{\rho,\varepsilon,V}| \; \ge \; s(p_V) V - c (\ln V)^2 \; \ge \; s(p_{\nu_V,0}) V - c' V^{1/4}
\end{equation}
for some $c,c' < + \infty$, and $p_V$ and $p_{\nu_V,0}$ defined in Lemma \ref{lem: high entropy state}. 

We will consider three different events. Let first 
\begin{equation*}
	A \; = \;  \left\{ \gamma \in \Gamma_{\rho,\varepsilon,V} \; : \;  \max_{x\in V} \gamma (x) \,\le\, (V \Delta \varepsilon)^{1/2} - V^{\alpha} \right\} 
\end{equation*}
with some $0 < \alpha < 1/2$ to be specified later. 
Let $I = \lfloor (V \Delta \varepsilon)^{1/2} - V^{\alpha} \rfloor$. 
By Lemma~\ref{lem: entropy gibbs on truncated space}, 
there exist $(\nu,\beta)\in \R^2$ such that the constraints \eqref{eq: density energy constraints for measures} are satisfied with $p = p_{\nu,\beta}^I$. 
Crucially, the Gibbs state $p_{\nu,\beta}^I$ has maximal entropy among all states satisfying the constraints \eqref{eq: density energy constraints for measures}
and being concentrated on $\{ 0 , \dots , I \}$. 
Hence, using Lemma~\ref{lem: number of states} to bound the number of states, we obtain 
\begin{equation}\label{eq: proof end 2}
	|A| \; \le \; \ed^{cV^{1/3}\ln V + s(p_{\nu,\beta}^I) V} 
\end{equation}
for some constant $c< + \infty$. 
By Lemma ~\ref{lem: entropy gibbs on truncated space} again,
$$
	s(p_{\nu,\beta}^I)
	\; = \; 
	s(p_{\nu_I,0}) + \mathcal O (I^{-2}\ln I)
	\; = \; 
	s(p_{\nu_I,0}) + \mathcal O (V^{-1}\ln V)
$$ 
Hence, inserting this estimate in \eqref{eq: proof end 2} and using \eqref{eq: proof end 1}, we get 
\begin{equation}\label{eq: estimate end 4}
	P_{\rho,e,V}^{\mathrm{mic}} (A)
	\; \le \; 
	\ed^{- (s(p_{\nu_V,0}) - s(p_{\nu_I,0})) V + \mathcal O (V^{1/3}\ln V)}
\end{equation}
An explicit computation shows that the entropy increases as a function of the density, i.e.\@ there exists $c> 0$ so that 
\begin{align*}
	s(p_{\nu_V,0}) - s(p_{\nu_I,0})
	\; &\ge \; 
	c (\rho_V - \rho_I)
	\; = \; 
	c \left( \frac{\Delta \varepsilon}{I} - \frac{\lfloor (V \Delta \varepsilon)^{1/2} \rfloor}{V} \right) \\
	\; &= \; 
	c \left( \frac{\Delta \varepsilon}{\lfloor (V \Delta \varepsilon)^{1/2} - V^\alpha\rfloor} -  \frac{\lfloor (V \Delta \varepsilon)^{1/2} \rfloor}{V}  \right)
	%\; \ge \; 
	%c \left( \frac{\Delta \varepsilon}{(V \Delta \varepsilon)^{1/2} - V^\alpha}  - \frac{(\Delta \varepsilon)^{1/2}}{V^{1/2}}\right) \\
	\; \ge \; 
	\frac{c'}{V^{1 - \alpha}} 
\end{align*}
for some $c' > 0$. 
Inserting this estimate in \eqref{eq: estimate end 4}, we find that $P_{\rho,e,V}^{\mathrm{mic}} (E) \to 0$ as $V \to \infty$, provided that $\alpha > 1/3$.

Next, let us consider the event 
$$
	B \; = \;  \left\{ \gamma \in \Gamma_{\rho,\varepsilon,V} \; : \;  \max_{x\in V} \gamma (x) \,\ge\, (V \Delta \varepsilon)^{1/2} + V^{\alpha} \right\} 
$$
with some $0 < \alpha < 1/2$ to be specified later. 
As before, we use Lemma~\ref{lem: number of states} to get the bound 
$$
	|B| \; \le \; \ed^{c V^{1/3} \ln V + \max_{p \in \tilde B} s(p)}
$$
for some $c< + \infty$, with $\tilde B = \{ p \in \mathcal P_{\rho,\varepsilon,V} : p = p(\gamma,\cdot) \text{ for some } \gamma \in B \}$
where $p(\gamma,\cdot)$ is defined by \eqref{eq: density of a configuration}. 
Given $p \in \tilde B$, there exists $n_0 \ge (V \Delta \varepsilon)^{1/2} + V^\alpha$ such that $p(n_0) \ge 1/V$. 
We decompose $p\in \tilde B$ as
$$
	p \; = \; \frac{1}{V} \delta_{n_0} + \left(1 - \frac{1}{V}\right) \tilde p
$$
and 
$$
	s(\tilde p) 
	\; = \;
	\frac{V}{V-1} s\left( p - \frac1V \delta_{n_0}\right) - \frac{V}{V-1} \ln \frac{V}{V-1}
	\; = \; 
	\frac{V}{V-1} s(p) + \mathcal O \left(\frac{\ln V}{V}\right) 
$$
and hence 
$
	s(p) \, \le \, s(\tilde p) +  \mathcal O (V^{-1}\ln V).
$
Now, $\sum_{n\ge 0}n \tilde p(n) \le \rho'_V$ with 
$$
	\rho'_V \; = \;  \rho - \frac{(V \Delta \varepsilon)^{1/2} + V^\alpha}{V}
$$
and therefore 
$
	s(\tilde p) \le s(p_{\nu'_V,0})
$
where $\nu'_V$ is such that $\rho_V' = \langle \gamma (0) \rangle_{\nu'_V,0}$, 
since the Gibbs states at infinite temperature maximize the entropy over the states with a given density,
and since the entropy of these states increases with the density. 
Hence 
$$
	|B| \; \le \; \ed^{s(p_{\nu'_V,0}) + c V^{1/3} \ln V}.
$$
Proceeding then as for the event $A$, we arrive at
	\begin{equation*}
	P_{\rho,e,V}^{\mathrm{mic}} (B)
	\; \le \; 
	\ed^{- (s(p_{\nu_V,0}) - s(p_{\nu'_V,0})) V + \mathcal O (V^{1/3}\ln V)}
\end{equation*}
and as before, we find that $P_{\rho,e,V}^{\mathrm{mic}} (B) \to 0$ as $V\to \infty$ provided that $\alpha > 1/3$. 

Finally, let us consider the event 
$$
	D \; = \;  \left\{ \gamma \in \Gamma_{\rho,\varepsilon,V} \; : \; \gamma_{\max}' \,\ge\, V^{\alpha'} \right\} 
$$
with some $0 < \alpha' < 1/2$ to be specified later.
Since $P_{\rho,e,V}^{\mathrm{mic}} (A) \to 0$ as $V \to \infty$, it suffices to bound the probability of $D \cap A^c$. 
We proceed through a very similar way as for the event $B$. 
This time, for $p \in \tilde{D \cap A^c}$, with 
$\tilde{D \cap A^c} = \{ p \in \mathcal P_{\rho,\varepsilon,V} : p = p(\gamma, \cdot) \text{ for some } \gamma \in D \cap A^c \}$,
we find either $n_0 \ge (V \Delta \varepsilon)^{1/2} - V^\alpha$ and $n_1 \ge V^{\alpha'}$ such that $p(n_0)\ge 1/V$ and $p(n_1)\ge 1/V$, 
or $n_0 \ge (V \Delta \varepsilon)^{1/2} - V^\alpha$ such that $p(n_0)\ge 2/V$ (the second case corresponding to the situation where $\gamma_{\max}'=\gamma_{\max}$). 
the rest of the proof is analogous, and one finds that $P_{\rho,e,V}^{\mathrm{mic}} (D) \to 0$ as $V\to \infty$ provided that $\alpha' > \alpha$. 

%$$
%	\rho'' \; = \;  \rho - \frac{(V \Delta \varepsilon)^{1/2} - V^\alpha + V^{\alpha'}}{V}, 
%$$

The theorem follows from the fact that the probability of the events $A$, $B$ and $D$ vanishes as $V\to \infty$, provided that $1/3 < \alpha < \alpha' < 1/2$.  
\end{proof}

\bigskip
\textbf{Acknowledgements.}
We thank N.~Starreveld for discussions at an early stage of this project, as well as C.~Bernardin and S.~Olla for pointing out relevant references to us. 
This work was partially supported by the grants ANR-15-CE40-0020-01 LSD and ANR-14-CE25-0011 EDNHS of the French National Research Agency (ANR).


\begin{thebibliography}{99}

\bibitem{abanin_et_al_2015}
D.~A.~Abanin, W.~De Roeck and F.~Huveneers,
Exponentially Slow Heating in Periodically Driven Many-Body Systems,
Physical Review Letters 
115 (25),
256803
(2015). 

\bibitem{bols_de_roeck_2018}
A.~Bols and W.~De Roeck, 
Asymptotic localization in the Bose-Hubbard model,
Journal of Mathematical Physics
59 (2), 
021901
(2018). 

\bibitem{callen}
H.~B.~Callen, 
Thermodynamics and an Introduction to Thermostatistics, 
Willey 
(1985). 

\bibitem{chatterjee_kirkpatrick_2012}
S.~Chatterjee and K.~Kirkpatrick, 
Probabilistic methods for discrete nonlinear Schrödinger equations,
Communications on Pure and Applied Mathematics
65 (5), 
727-757 
(2012).

\bibitem{chatterjee_2014}
S.~Chatterjee,
Invariant measures and the soliton resolution conjecture,
Communications on Pure and Applied Mathematics
67 (11), 
1737-1842 
(2014).

\bibitem{chatterjee_2017}
S.~Chatterjee, 
A note about the uniform distribution on the intersection of a simplex and a sphere. 
Journal of Topology and Analysis
9 (04), 
717-738
(2017).

\bibitem{chelboun_grosskinsky_2014}
P.~Chleboun and S.~Grosskinsky, 
Condensation in stochastic particle systems with stationary product measures, 
Journal of Statistical Physics
154 (1-2), 
432-465 
(2014).

\bibitem{cherny_et_al_2019}
A.~Yu.~Cherny, T.~Engl and S.~Flach,
Non-Gibbs states on a Bose-Hubbard lattice,
Physical Review A
99,
023603
(2019).

\bibitem{deutsch_1991}
J.~M.~Deutsch,
Quantum statistical mechanics in a closed system, 
Physical Review A
43 (4), 
2046
(1991).

\bibitem{de_roeck_huveneers_2014}
W.~De Roeck and F.~Huveneers,
Asymptotic Quantum Many-Body Localization from Thermal Disorder,
Communications in Mathematical Physics
332 (3), 
1017-1082
(2014). 

\bibitem{de_roeck_huveneers_2015} 
W.~De Roeck and F.~Huveneers,
Asymptotic Localization of Energy in Nondisordered Oscillator Chains,
Communications on Pure and Applied Mathematics
68 (9),
1532-1568
(2015).

\bibitem{evans_hanney_2005}
M.~R.~Evans and T.~Hanney, 
Nonequilibrium statistical mechanics of the zero-range process and related models, 
Journal of Physics A: Mathematical and General
38 (19),
R195-R240
(2005). 

\bibitem{grosskinsky_2008}
S.~Grosskinsky, 
Equivalence of ensembles for two-species zero-range invariant measures,
Stochastic Processes and their Applications
118 (8),
1322-1350
(2008). 

\bibitem{grosskinsky_et_al_2003}
S.~Grosskinsky, G.~Sch{\"u}tz and H.~Spohn,
Condensation in the zero range process: stationary and dynamical properties,
Journal of Statistical Physics
113 (3-4), 
389-410
(2003). 

\bibitem{iubini_et_al_2013}
S.~Iubini, R.~Franzosi, R.~Livi, G.-L.~Oppo and A.~Politi,
Discrete breathers and negative-temperature states,
New Journal of Physics
15 (2),
023032
(2013).

\bibitem{lieb_et_al_2005}
E.~H.~Lieb, R.~Seiringer, J~P.~Solovej and J.~Yngvason, 
The mathematics of the Bose gas and its condensation, 
Oberwolfach Seminars, 
Vol.~34,
Birkh\"auser Verlag
(2005).

%\bibitem{lukkarinen_2018}
%J.~Lukkarinen,
%Multi-state condensation in Berlin-Kac spherical models,
%arXiv e-prints,
%arXiv:1806.01806
%(2018).

\bibitem{majumdar_2010}
S.~N.~Majumdar, 
Real-space Condensation in Stochastic Mass Transport Models,
Exact Methods in Low-dimensional Statistical Physics and Quantum Computing: Lecture Notes of the Les Houches Summer School: Volume 89, July 2008,
Oxford University Press
(2010).

\bibitem{mithun_et_al_2018}
T.~Mithun, Y.~Kati, C.~Danieli and S.~Flach, 
Weakly Nonergodic Dynamics in the Gross-Pitaevskii Lattice,
Physical Review Letters
120 (18),
184101
(2018). 

\bibitem{nam_2018}
K.~Nam,
Large deviations and localization of the microcanonical ensembles given by multiple constraints,
arXiv e-prints, 
arXiv:1809.04138
(2018). 

\bibitem{rasmussen_et_al_2000}
K.~\O{}.~Rasmussen, T.~Cretegny, P.~G.~Kevrekidis and N.~Gr\o{}nbech-Jensen,
Statistical Mechanics of a Discrete Nonlinear System,
Physical Review Letters
84 (17),
3740-3743
(2000). 

\bibitem{rumpf_2004}
B.~Rumpf, 
Simple statistical explanation for the localization of energy in nonlinear lattices with two conserved quantities,
Physical Review E
69, 
016618
(2004).

\bibitem{srednicki_1994}
M.~Srednicki, 
Chaos and quantum thermalization,
Physical Review E
50 (2), 
888
(1994). 

\bibitem{touchette_2011}
H.~Touchette,
Ensemble equivalence for general many-body systems,
Europhysics Letters
96 (5),
50010
(2011).

\bibitem{touchette_2015}
H.~Touchette,
Equivalence and Nonequivalence of Ensembles: Thermodynamic, Macrostate, and Measure Levels,
Journal of Statistical Physics
159 (5),
987-1016
(2015).


\end{thebibliography}
\end{document}